\documentclass[11pt,a4paper]{article}
\usepackage[T1]{fontenc}
\usepackage[bitstream-charter]{mathdesign}

\usepackage[latin1]{inputenc}							
\usepackage{amsmath}									
\usepackage{amsthm}
\usepackage{amsbsy}
\usepackage[math]{easyeqn}
\usepackage{esint}
\usepackage{graphicx}
\usepackage{xcolor}
\definecolor{bl}{rgb}{0.0,0.2,0.6} 

\usepackage[authoryear]{natbib}
\usepackage{comment}

\usepackage{sectsty}
\usepackage[compact]{titlesec} 
\allsectionsfont{\color{bl}\scshape\selectfont}

\numberwithin{equation}{section}
\numberwithin{figure}{section}
\theoremstyle{plain}
\newtheorem{thm}{\theoremname}
\theoremstyle{definition}
\newtheorem{example}[thm]{\examplename}
\theoremstyle{remark}
\newtheorem{rem}[thm]{\remarkname}
\theoremstyle{plain}
\newtheorem{lem}[thm]{\lemmaname}

\providecommand{\lemmaname}{Lemma}
  \providecommand{\remarkname}{Remark}
  \providecommand{\examplename}{Example}
\providecommand{\theoremname}{Theorem}

\newcommand{\arginf}{\operatornamewithlimits{arginf}}

\makeatletter							
\def\printtitle{
    {\color{bl} \centering \huge \sc \textbf{\@title}\par}}		
\makeatother							

\title{Pricing Bermudan options via multi-level approximation methods 
\vspace{10pt}
\footnote{\footnotesize This research was partially supported by the Deutsche
      Forschungsgemeinschaft through the SPP 1324 ``Mathematical methods for extracting quantifiable     information from complex systems'' and   by
Laboratory for Structural Methods of Data Analysis in Predictive Modeling, MIPT, RF government grant, ag. 11.G34.31.0073.}
}

\makeatletter							
\def\printauthor{
    {\centering \small \@author}}				
\makeatother							

\author{%
	Denis Belomestny, Fabian Dickmann and Tigran Nagapetyan \\
	Duisburg-Essen University and Fraunhofer ITWM \\
	\vspace{20pt}
	}

\usepackage{fancyhdr}
\usepackage{lastpage}	
\usepackage[runin]{abstract}			
\abslabeldelim{\quad}						%
\begin{document}
\printtitle

\printauthor
\begin{abstract}
In this article we propose a novel approach to reduce the computational
complexity of various approximation methods for pricing discrete
time American or Bermudan options. Given a sequence of continuation values estimates
corresponding to different levels of spatial approximation, we propose a multi-level low biased estimate for the
price of the option. It turns out that the resulting complexity
gain can be of order \ensuremath{\varepsilon^{-1}}
 with \ensuremath{\varepsilon}
 denoting the desired precision. The performance of the proposed multilevel
algorithms is illustrated by a numerical example. 
\end{abstract}

\section{Introduction}
Pricing of an American option usually reduces to solving an optimal stopping
problem that can be efficiently solved in low dimensions via dynamic
programming algorithm. However, many problems arising in practice
(see e.g. \citet{Gl}) have high dimensions, and these applications
have motivated the development of Monte Carlo methods for pricing
American option. Pricing American style derivatives via Monte Carlo
is a challenging task, because it requires the backwards dynamic programming
algorithm that seems to be incompatible with the forward structure
of Monte Carlo methods. In recent years much research was focused on the development
of fast methods to compute approximations to the optimal exercise
policy. Eminent examples include the functional optimization approach
of \citet{A}, the mesh method of \citet{BG}, the regression-based
approaches of \citet{Car}, \citet{LS}, \citet{TV}, \citet{E} and
\citet{B1}. The complexity of the fast approximations algorithms
depends on the desired precision \ensuremath{\varepsilon}
 in a quite nonlinear way that, in turn, is determined by some fine
properties of the underlying exercise boundary and the continuation
values (see, e.g., \citet{B1}). In some situations (e.g. in the case
of the stochastic mesh method or local regression) this complexity
is of order \ensuremath{\varepsilon^{-3}},
 which is rather high. One way to reduce the complexity of the fast
approximation methods is to use various variance reduction methods.
However, the latter methods are often ad hoc and, more importantly,
do not lead to provably reduced asymptotic complexity. In this paper
we propose a generic approach which is able to reduce the order of
asymptotic complexity  and which is applicable to various fast
approximation methods, such as global regression, local regression
or stochastic mesh method. The main idea of the method is inspired
by the pathbreaking work of \citet{Gi} that introduced a multilevel
idea into stochastics. As  similar to
the recent work of \citet{BSD}, we consider not  levels corresponding
to different discretization steps, but  levels related to different
degrees of approximation of the continuation values. For example,
in the case of the Longstaff-Schwartz algorithm, the latter degree is basically governed
by the number of basis functions and in the case of the mesh method by the number of training paths used to approximate the continuation values. The new multi-level approach is able
to significantly reduce the complexity of the fast approximation methods
leading in some cases to the complexity gain of order \ensuremath{\varepsilon^{-1}.}
 The paper is organised as follows. In Section~\ref{mainsetup} the pricing problem is formulated, the main assumptions are introduced and illustrated.  In Section~\ref{compl_anal} the complexity analysis of a generic approximation algorithm is carried out.  The main multi-level Monte Carlo algorithm is introduced in  Section~\ref{compl_anal_ml} where also its complexity is studied. In Section~\ref{num_example} we numerically test our approach for the problem of pricing  Bermudan max-call options via mesh method. The proofs are collected in Section~\ref{proofs}.

\section{Main setup}
\label{mainsetup}
An American option grants the holder the right to select the time
at which to exercise the option, and in this differs from a European
option that may be exercised only at a fixed date. A general class
of American option pricing problems can be formulated through an \ensuremath{\mathbb{R}^{d}}
 Markov process \ensuremath{\{X_{t},\,0\leq t\leq T\}}
 defined on a filtered probability space \ensuremath{(\Omega,\mathcal{F},(\mathcal{F}_{t})_{0\leq t\leq T},\mathrm{P})}
. It is assumed that the process \ensuremath{(X_t)}
 is adapted to \ensuremath{(\mathcal{F}_{t})_{0\leq t\leq T}}
 in the sense that each \ensuremath{X_{t}}
 is \ensuremath{\mathcal{F}_{t}}
 measurable. Recall that each \ensuremath{\mathcal{F}_{t}}
 is a \ensuremath{\sigma}
-algebra of subsets of \ensuremath{\Omega}
 such that \ensuremath{\mathcal{F}_{s}\subseteq\mathcal{F}_{t}\subseteq\mathcal{F}}
 for \ensuremath{s\leq t.} We restrict attention to options admitting a finite set of exercise
opportunities \ensuremath{0=t_{0}<t_{1}<t_{2}<\ldots<t_{\mathcal{J}}=T,}  called Bermudan options. Then 
\[
Z_{j}:=X_{t_{j}},\quad j=0,\ldots,\mathcal{J},
\]
is a Markov chain. If exercised at time \ensuremath{t_{j},\, j=1,\ldots,\mathcal{J}}, the option pays \ensuremath{g_{j}(Z_{j})}, for some known functions \ensuremath{g_{0},g_{1},\ldots,g_{\mathcal{J}}}
 mapping \ensuremath{\mathbb{R}^{d}}
 into \ensuremath{[0,\infty)}. Let \ensuremath{\mathcal{T}_{j}}
 denote the set of stopping times taking values in \ensuremath{\{j,j+1,\ldots,\mathcal{J}\}}. A standard result in the theory of contingent claims states that
the equilibrium price \ensuremath{V_{j}(z)}
 of the Bermudan option at time \ensuremath{t_{j}}
 in state \ensuremath{z}
 given that the option was not exercised prior to \ensuremath{t_{j}}
 is its value under an optimal exercise policy: 
\begin{eqnarray*}
V_{j}^{*}(z)=\sup_{\tau\in\mathcal{T}_{j}}\mathrm{E}[g_{\tau}(Z_{\tau})|Z_{j}=z],\quad z\in\mathbb{R}^{d}.
\end{eqnarray*}
A common feature of all fast approximation algorithms is that they
deliver estimates \ensuremath{C_{k,0}(z),\ldots,C_{k,\mathcal{J}-1}(z)}
 for the so called continuation values: 
\begin{eqnarray}
C_{j}^{*}(z):=\mathrm{E}[V_{j+1}^{*}(Z_{j+1})|Z_{j}=z],\quad j=0,\ldots,\mathcal{J}-1,\label{CV}
\end{eqnarray}
based on the set of trajectories \ensuremath{(Z_{0}^{(i)},\ldots,Z_{\mathcal{J}}^{(i)}),}
 \ensuremath{i=1,\ldots,k,}
 all starting from one point, i.e., \ensuremath{Z_{0}^{(1)}=\ldots=Z_{0}^{(k)}.}
 In the case of the so-called regression methods and the mesh method, the estimates for
the continuation values are obtained via the recursion (\textit{dynamic programming
principle}): 
\begin{eqnarray*}
C_{\mathcal{J}}^{*}(z) & = & 0,\\
C_{j}^{*}(z) & = & \mathrm{E}[\max(g_{j+1}(Z_{j+1}),C_{j+1}^{*}(Z_{j+1}))|Z_{j}=z]
\end{eqnarray*}
combined with Monte Carlo: at \ensuremath{(\mathcal{J}-j)}
th step one estimates the expectation
\begin{equation}
\mathrm{E}[\max(g_{j+1}(Z_{j+1}),C_{k,j+1}(Z_{j+1}))|Z_{j}=z]\label{regr_aim}
\end{equation}
via regression (global or local) based on the set of paths 
\[
(Z_{j}^{(i)},C_{k,j+1}(Z_{j+1}^{(i)})),\quad i=1,\ldots,k,
\]
where \ensuremath{C_{k,j+1}(z)}
 is the estimate for \ensuremath{C_{j+1}^{*}(z)}
 obtained in the previous step. 
 \par
 Based on the estimates \ensuremath{C_{k,0}(z),\ldots,C_{k,\mathcal{J}-1}(z)}
 we can construct a lower bound (low biased estimate) for \ensuremath{V_{0}^{*}}
 using the (generally suboptimal) stopping rule: 
\begin{eqnarray*}
\tau_{k}=\min\{0\leq j\leq \mathcal{J}:g_{j}(Z_{j}) \geq C_{k,j}(Z_{j})\}
\end{eqnarray*}
with \ensuremath{C_{k,\mathcal{J}}\equiv0}
by definition. Fix now a natural number \ensuremath{n}
and simulate \ensuremath{n} new independent
 trajectories of the process \ensuremath{Z.}
 A low-biased estimate for \ensuremath{V_{0}^{*}}
 can be then defined as 
\begin{eqnarray}
\label{mc}
V_{0}^{n,k}=\frac{1}{n}\sum_{r=1}^{n}g_{\tau_{k}^{(r)}}(Z_{\tau_{k}^{(r)}}^{(r)})
\end{eqnarray}
with
\[
\tau_{k}^{(r)}=\inf\{0\leq j\leq\mathcal{J}:g_{j}(Z_{j}^{(r)})\geq C_{k,j}(Z_{j}^{(r)})\}.
\]
Thus any fast approximation  approximation algorithm can be viewed as consisting of the following two steps.
\begin{description}
\item [Step 1] Construction of the estimates \( C_{k,j}, \) \( j=1,\ldots,J, \) on \( k \) training paths. 
\item [Step 2] Construction of the low-biased estimate \( V_{0}^{n,k} \) by evaluating functions \(  C_{k,j}, \) \( j=1,\ldots, J, \) on each of new \( n \) testing trajectories.
\end{description}
\par
Let us now consider a generic family of the continuation values estimates \ensuremath{C_{k,0}(z),\ldots,C_{k,\mathcal{J}-1}(z)}
 with the natural number \ensuremath{k}
 determining the quality of the estimates as well as their complexity.
In particular we make the following assumptions. 
\begin{description}
\item [{(AP)}] For any \ensuremath{k\in\mathbb{N}}
 the estimates \ensuremath{C_{k,0}(z),\ldots,C_{k,\mathcal{J}-1}(z)}
 are defined on some probability space \ensuremath{(\Omega^{k},\mathcal{F}^{k},\mathrm{P}^{k})}
 which is independent of \ensuremath{(\Omega,\mathcal{F},\mathrm{P}).}
 \item [{(AC)}] For any \ensuremath{j=1,\ldots,\mathcal{J},}
 the cost of constructing the estimate \( C_{k,j} \) on \(k\) training paths, i.e., \( C_{k,j}\bigl(Z_j^{(i)}\bigr), \)  \( i=1,\ldots, k, \)  is of order \( k\times k^{\varkappa_1} \) for some \ensuremath{\varkappa_1>0}
 and 
 the cost of evaluating $C_{k,j}(z)$ in a new point \( z\not\in \{Z_j^{(1)},\ldots, Z_j^{(k)}\} \)   is of order $k^{\varkappa_2}$
for some \( \varkappa_2>0. \)
 \end{description}

\begin{description}
\item [{(AQ)}] There is a sequence of positive real numbers \ensuremath{\gamma_{k}}
 with \ensuremath{\gamma_{k}\rightarrow0,}
 \ensuremath{k\to\infty}
 such that 
\[
\mathrm{P}^{k}\left(\sup_{z}\left\vert C_{k,j}(z)-C_{j}^{\ast}(z)\right\vert >\eta\sqrt{\gamma_{k}}\right)<B_{1}e^{-B_{2}\eta},\quad\eta>0
\]
for some constants \ensuremath{B_{1}>0}
 and \ensuremath{B_{2}>0.}
\end{description}
\paragraph{Discussion}
\begin{itemize}
\item Given (AC) the overall complexity of a fast approximation algorithm is proportional to
\begin{equation}
\label{compl_all}
k^{1+\varkappa_1}+n\times k^{\varkappa_2},
\end{equation}
where the first term in \eqref{compl_all} represents the cost of constructing  the estimates \( C_{k,j}, \) \( j=1,\ldots,J, \) on training paths and the second one gives the cost of evaluating the estimated  continuation values on \(n \) testing paths. 
\item Additionally, one usually has to take into account the cost of paths simulation. If the process  \( X \) solves a stochastic differential equation and the Euler discretisation scheme with  time step \( h \) is used to generate paths, then the term \(k\times h^{-1}+ n\times h^{-1}\) needs to be added to   \eqref{compl_all}. In order to make the analysis more focused and transparent we do not
take here the path generation costs into account.  
\end{itemize}
Let us now illustrate the above assumptions for three well known fast approximation methods. 
\begin{example}[Global regression] 
Fix a vector of real-valued functions \ensuremath{\psi=(\psi_{1},\ldots,\psi_{M})}
 on \ensuremath{\mathbb{R}^{d}.} Suppose that the estimate \(C_{k,j+1}\) is already constructed and has the form
\begin{equation*}
C_{k,j+1}(z)=\alpha_{j+1,1}^{k}\psi_{1}(z)+\ldots+\alpha_{j+1,M}^{k}\psi_{L}(z)
\end{equation*}
for some \((\alpha_{j+1,1}^{k},\ldots, \alpha_{j+1,M}^{k})\in \mathbb{R}^M.\)
Let \ensuremath{\boldsymbol{\alpha}_{j}^{k}=(\alpha_{j,1}^{k},\ldots,\alpha_{j,M}^{k})}
 be a solution of the following least squares optimization problem:
\begin{equation}
\arginf_{\boldsymbol{\alpha}\in\mathbb{R}^{M}}\sum_{i=1}^{k}\left[\zeta_{j+1,k}(Z_{j+1}^{(i)})-\alpha_{1}\psi_{1}(Z_{j}^{(i)})-\ldots-\alpha_{M}\psi_{M}(Z_{j}^{(i)})\right]^{2}\label{eq:glob_ls}
\end{equation}
with \ensuremath{\zeta_{j+1,k}(z)=\max\left\{ g_{j+1}(z),C_{k,j+1}(z)\right\},} where \(C_{k,j+1}\) is the estimate of \(C_{j+1}^{*}\) obtained in the previous step.
 Define the approximation for \ensuremath{C_{j}^{*}}
 via 
\[
C_{k,j}(z)=\alpha_{j,1}^{k}\psi_{1}(z)+\ldots+\alpha_{j,M}^{k}\psi_{M}(z),\quad z\in\mathbb{R}^{d}.
\]
It is clear that all estimates \ensuremath{C_{k,j}}
 are well defined on the cartesian product of \ensuremath{k}
 independent copies of \ensuremath{(\Omega,\mathcal{F},\mathrm{P}).}
 The complexity $\mathrm{comp}(\boldsymbol{\alpha}_{j}^{k})$ of computing
$\boldsymbol{\alpha}_{j}^{k}$ is of order \ensuremath{k\cdot M^{2}+\mathrm{comp}(\boldsymbol{\alpha}_{j+1}^{k}),}
 since each \ensuremath{\boldsymbol{\alpha}_{j}^{k}}
 is of the form \ensuremath{\boldsymbol{\alpha}_{j}^{k}=B^{-1}b}
 with 
\[
B_{p,q}=\frac{1}{k}\sum_{i=1}^{k}\psi_{p}(Z_{j}^{(i)})\psi_{q}(Z_{j}^{(i)})
\]
and 
\[
b_{p}=\frac{1}{k}\sum_{i=1}^{k}\psi_{p}(Z_{j}^{(i)})\zeta_{k,j+1}(Z_{j+1}^{(i)}),
\]
\ensuremath{p,q\in\{1,\ldots,M\}.}
Iterating backwardly in time we get $\mathrm{comp}(\boldsymbol{\alpha}_{j}^{k})\sim(\mathcal{J}-j)\cdot k\cdot M^{2}.$
Furthermore, it can be shown that the estimates \ensuremath{C_{k,0}(z),\ldots,C_{k,\mathcal{J}-1}(z)}
 satisfy the assumption (AQ) with \( \gamma_k=1/k, \)  provided \ensuremath{M}
 increases with \ensuremath{k} at polynomial rate, i.e., \( M=k^{\rho} \) for some \( \rho>0\) (see,
e.g., \citet{Z}). 
 Thus, the parameters \( \varkappa_1 \) and \( \varkappa_2 \) in (AC) are given by \( 2\rho \) and \( \rho, \)
 respectively.
\end{example}
\begin{example}[Local regression] Local polynomial regression estimates can be
defined as follows. Fix some \ensuremath{j}
 such that \ensuremath{0\leq j<\mathcal{J}}
 and suppose that we want to compute the expectation in \eqref{regr_aim}:
\begin{eqnarray*}
\mathrm{E}[\zeta_{j+1,k}(Z_{j+1})|Z_{j}=z],\quad z\in\mathbb{R}^{d}
\end{eqnarray*}
with \ensuremath{\zeta_{j+1,k}(z)=\max\left\{ g_{j+1}(z),C_{k,j+1}(z)\right\} .}
 For some \ensuremath{\delta>0}, \ensuremath{z\in\mathbb{R}^{d}}, an integer \ensuremath{l\geq0}
 and a function \ensuremath{K:\mathbb{R}^{d}\to\mathbb{R}_{+}}, denote by \ensuremath{q_{z,k}}
 a polynomial on \ensuremath{\mathbb{R}^{d}}
 of degree \ensuremath{l} (i.e. the maximal order of the multi-index is less than or equal to \ensuremath{l}) which minimizes 
\begin{equation}
\sum_{i=1}^{k}\left[\zeta_{j+1,k}(Z_{j+1}^{(i)})-q(Z_{j}^{(i)}-z)\right]^{2}K\left(\frac{Z_{j}^{(i)}-z}{\delta}\right)\label{OF}
\end{equation}
over the set of all polynomials \( q \) of degree \( l. \)
The local polynomial estimator of order \ensuremath{l}
 for  \ensuremath{C_{j}^{*}(z)}
 is then defined as \ensuremath{C_{k,j}(z)=q_{z,k}(0)}
 if \ensuremath{q_{z,k}}
 is the unique minimizer of \eqref{OF} and \ensuremath{C_{k,j}(z)=0}
 otherwise. 
 The value \ensuremath{\delta}
 is called a bandwidth and the function \ensuremath{K}
 is called a kernel function. In \citet{B1} it is shown that the
local polynomial estimates \ensuremath{C_{k,0}(z),\ldots,C_{k,\mathcal{J}-1}(z)}
 of degree \ensuremath{l}
 satisfy the assumption (AQ) with \( \gamma_k= k^{-2\beta/(2\beta+d)} \)  under \( \beta \)-H\"older smoothness of the continuation values \ensuremath{C^*_{0}(z),\ldots,C^{*}_{\mathcal{J}-1}(z)}, provided
\ensuremath{\delta=k^{-1/(2l+d)}.} Since in general the summation in \eqref{OF} runs over all \( k \) paths (see Figure \ref{fig_continuation}) we have \( \varkappa_1=1 \) and \( \varkappa_2=1 \)  in (AC).
 \end{example}
 \begin{figure}[tbh]
\centering
\includegraphics[width=0.75\textwidth]{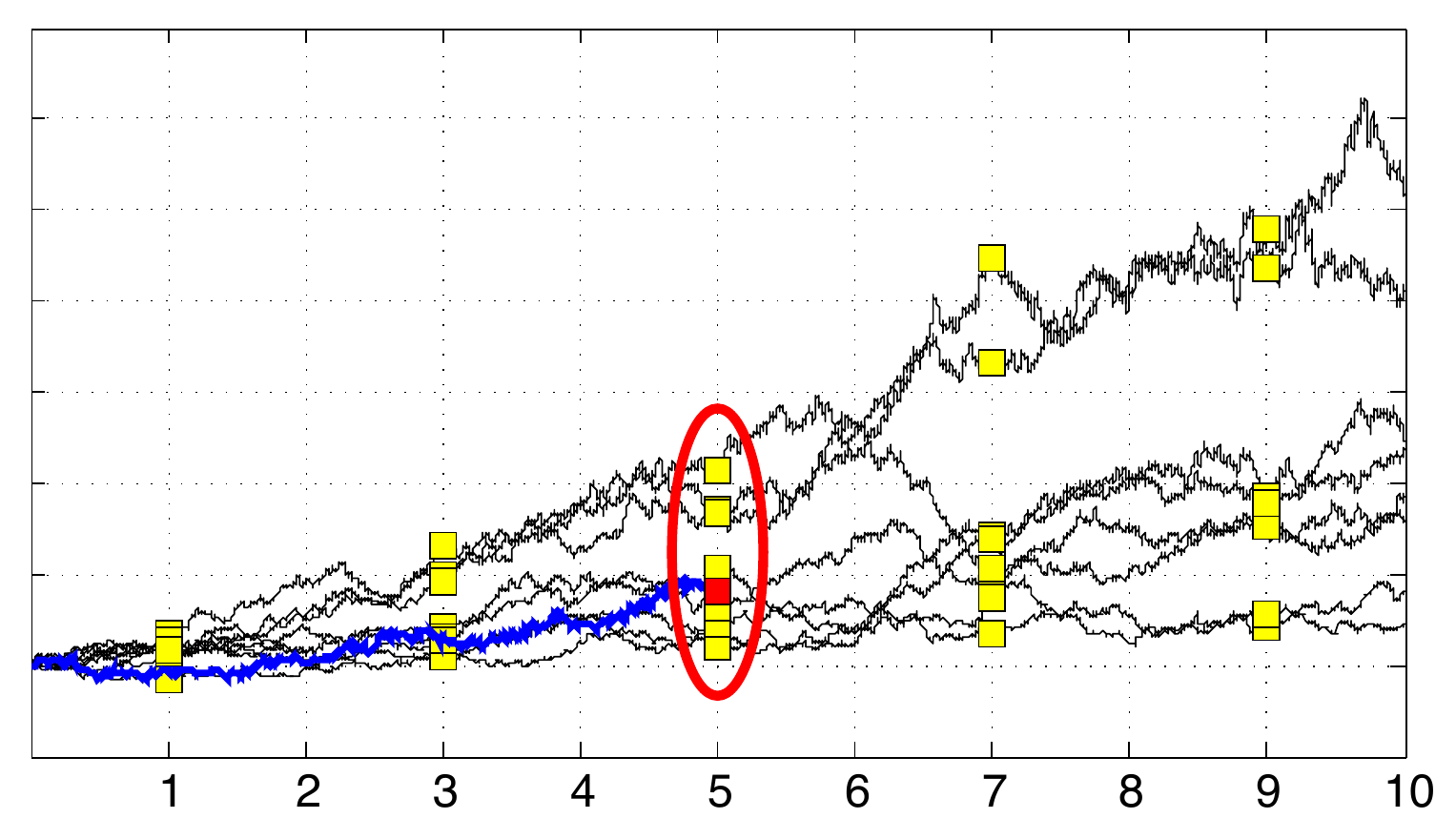}
\caption{Local regression and mesh methods: in order to compute the continuation value estimate  \( C_{k,5} \) in a point (red) lying on a testing path (blue), all \( k \) points (yellow) on training paths at time \( 5 \) have to be used. }
\label{fig_continuation}
\end{figure}
\begin{example}[Mesh Method]
\label{mesh_ex}
In the mesh method of \citet{BG4} the continuation value \( C^*_j \) at a point \( z \) is approximated via
\begin{eqnarray}
\label{mesh_est}
C_{k,j}(z)=\frac{1}{k}\sum_{i=1}^{k}\zeta_{k,j+1}(Z_{j+1}^{(i)})\cdot w_{ij}(z),
\end{eqnarray}
 where  \( \zeta_{k,j+1}(z)=\max\left\{ g_{j+1}(z),C_{k,j+1}(z)\right\} \) and 
\begin{eqnarray*}
w_{ij}(z)=\frac{p_j(z,Z_{j+1}^{(i)})}{\frac{1}{k}\sum_{l=1}^{k} p_j(Z_{j}^{(l)},Z_{j+1}^{(i)})},
\end{eqnarray*}
where \( p_j(x,\cdot) \) is the conditional density of  \( Z_{j+1} \) given \(Z_j=x.\)
Again the summation in \eqref{mesh_est} runs over all \( k \) paths. Hence \( \varkappa_1=1 \) in (AC) and for any \( j=0,\ldots, \mathcal{J}-1, \) the complexity of computing \ensuremath{C_{k,j}(z)}
in a point \( z \) not belonging to the set of training trajectories is of order \ensuremath{k} (see Figure \ref{fig_continuation}), provided the transition density  \( p_j(x,y) \) is analytically known. For assumption (AQ) see, e.g., 
\citet{AJ}.
\end{example}

 \section{Complexity analysis of \ensuremath{V_{0}^{n,k}}
}
\label{compl_anal}
We shall use throughout the notation $A\lesssim B$ if $A$ is
bounded by a constant multiple of $B$, independently of the
parameters involved, that is, in the Landau notation $A=O(B)$.
Equally $A\gtrsim B$ means $B\lesssim A$ and $A\thicksim B$ stands
for $A\lesssim B$ and $A\gtrsim B$ simultaneously.
\par
In order to carry out the complexity analysis of the estimate \eqref{mc}
we need the so-called ``margin'' or boundary assumption. 
\begin{description}
\item [{(AM)}] There exist constants $A>0$, \ensuremath{\delta_{0}>0}
 and $\alpha>0$ such that
\[
\mathrm{P}\left(|C_{j}^{*}(Z_{j})-g_j(Z_{j})|\leq\delta\right)\leq A\delta^{\alpha}
\]
for all \ensuremath{j=0,\ldots,\mathcal{J},}
 and all \ensuremath{\delta<\delta_{0}.}
 
\end{description}
\begin{rem}
Assumption (AM) provides a useful characterization of the behavior
of the continuation values $(C_{j}^{*})$ and payoffs \ensuremath{(g_{j})}
 near the exercise boundary \ensuremath{\partial\mathcal{E}}
 with 
\begin{eqnarray*}
\mathcal{E}=\left\{ (j,x):g_{j}(x)\geq C_{j}^{*}(x)\right\} .
\end{eqnarray*}
In the situation when all functions \ensuremath{C_{j}^{*}-g_{j},\, j=0,\ldots,\mathcal{J}-1,}
 are smooth and have non-vanishing derivatives in the vicinity of
the exercise boundary, we have \ensuremath{\alpha=1}. Other values of \ensuremath{\alpha}
 are possible as well, see \citet{B1}. 
\end{rem} 
Let us now turn to the properties of the estimate \(V_{0}^{n,k}.\) While the variance of the
estimate \ensuremath{V_{0}^{n,k}}
 is given by 
\begin{equation}
\label{eq:var_smc}
\operatorname{Var}[V_{0}^{n,k}]=\operatorname{Var}[g_{\tau_{k}}(Z_{\tau_{k}})]/n,
\end{equation}
its bias is analyzed in the following theorem. 
\begin{thm}
\label{thm:bias}Suppose that (AP), (AM) and (AQ) hold with some \ensuremath{\alpha>0}, and all functions \ensuremath{g_{j}}
 are uniformly bounded, i.e., 
\[
|g_{j}(x)|\leq G, \quad x \in\mathbb{R}^{d}.
\]
for some constants \ensuremath{G>0}.
Then it holds 
\[
\left|V_{0}^{*}-\mathrm{E}[V_{0}^{n,k}]\right|\lesssim\gamma_{k}^{(1+\alpha)/2},\quad k\to\infty.
\]
\end{thm}
The next theorem gives an upper estimate for the complexity of \ensuremath{V_{0}^{n,k}.}
\begin{thm}
\label{compl_mc} Let assumptions (AP), (AC), (AQ) and (AM)  hold with
\[
\gamma_{k}=k^{-\mu},\quad k\in\mathbb{N}
\]
for some \ensuremath{\mu>0.}
Then for any  \( \varepsilon>0 \)  the choice $$k^* = \varepsilon^{-\frac{2}{\mu(1+\alpha)}}, \quad n^*=\varepsilon^{-2}$$ leads to
\[
\mathrm{E}\left[V_{0}^{n^*,k^*}-V_{0}^{*}\right]^{2}\leq\varepsilon^{2},
\]
and the complexity of the estimate \( V_{0}^{n^*,k^*} \) is bounded from above by  \ensuremath{\mathcal{C}_{n^*,k^*}(\varepsilon)} with 
\begin{equation}
\label{eq:comp_smc}
\mathcal{C}_{n^*,k^*}(\varepsilon)\lesssim\varepsilon^{-2\cdot\max\left(\frac{\varkappa_1+1}{\mu(1+\alpha)},\, 1+\frac{\varkappa_2}{\mu(1+\alpha)}\right)}, \quad \varepsilon\to 0.
\end{equation}
\end{thm}

\paragraph{Discussion}
Theorem~\ref{compl_mc} implies that  the complexity of  \( V_{0}^{n^*,k^*} \) is always larger than \( \varepsilon^{-2}. \) In the case  \( \varkappa_1=1 \) and \( \varkappa_2=1 \)  (mesh method or local regression) we get 
\begin{equation}
\label{mc_compl_mesh}
\mathcal{C}_{n^*,k^*}(\varepsilon)\lesssim\varepsilon^{-2\max\left(\frac{2}{\mu(1+\alpha)},1+\frac{1}{\mu(1+\alpha)}\right)}
\end{equation}
Furthermore, in the most common case  \(  \alpha=1 \) the bound \eqref{mc_compl_mesh} simplifies to 
\begin{eqnarray*}
\mathcal{C}_{n^*,k^*}(\varepsilon)\lesssim\varepsilon^{-2\max\left(\frac{1}{\mu},1+\frac{1}{2\mu}\right)}.
\end{eqnarray*}
Since for all regression methods and the mesh method \(\mu\leq 1,\) the asymptotic complexity is always larger 
than \(\varepsilon^{-3}.\) In the next section we present a multilevel approach that can reduce the asymptotic complexity down to \(\varepsilon^{-2}\) in some cases.
\section{Multilevel approach}
\label{compl_anal_ml}
Fix some natural number \ensuremath{L}
 and let \ensuremath{\mathbf{k}=(k_0,k_{1},\ldots,k_{L})}
 and \ensuremath{\mathbf{n}=(n_0, n_{1},\ldots,n_{L})}
 be two sequences of natural numbers, satisfying \(k_0<k_1<\ldots<k_L\) and \(n_0>n_1>\ldots>n_L.\) Define 
\[
V_{0}^{\mathbf{n},\mathbf{k}}=\frac{1}{n_{0}}\sum_{r=1}^{n_{0}}g_{\tau^{(r)}_{k_{0}}}\Bigl(Z_{\tau^{(r)}_{k_{0}}}^{(r)}\Bigr)+\sum_{l=1}^{L}\frac{1}{n_{l}}\sum_{r=1}^{n_{l}}\left[g_{\tau^{(r)}_{k_{l}}}\Bigl(Z_{\tau^{(r)}_{k_{l}}}^{(r)}\Bigr)-g_{\tau^{(r)}_{k_{l-1}}}\Bigl(Z_{\tau^{(r)}_{k_{l-1}}}^{(r)}\Bigr)\right]
\]
with
\[
\tau_{k}^{(r)}=\inf\left\{0\leq j\leq\mathcal{J}:g_{j}({Z}_{j}^{(r)})\geq C_{k,j}({Z}_{j}^{(r)})\right\}, \quad k\in \mathbb{N},
\]
where for any \(l=1,\ldots, L,\) both estimates  \(C_{k_l,j}\) and \(C_{k_{l-1},j}\) are based on one  set of \(k_l\) training trajectories.
Let us analyse the properties of the estimate \(V_{0}^{\mathbf{n},\mathbf{k}}.\) First note that its bias   coincides with the bias of \(g_{\tau_{k_{L}}}(Z_{\tau_{k_{L}}})\) corresponding to the finest approximation level. As to the variance of \(V_{0}^{\mathbf{n},\mathbf{k}},\) it can be significantly reduced due  the use of ``good'' continuation value estimates  \(C_{k_{l-1},j}\) and \(C_{k_l,j}\) (that are both close to \(C^*_{j}\)) on the same set of testing trajectories in each level. In this way a ``coupling'' effect is achieved. 
The following theorem quantifies the above heuristics.  
\begin{thm}
\label{biasvarml}
Let (AP), (AQ) and (AM) hold with some \ensuremath{\alpha>0,}
 then the estimate \ensuremath{ V_{0}^{\mathbf{n},\mathbf{k}}}
 has the bias of order  \(\gamma_{k_L}^{(1+\alpha)/2}\)
 and the variance of order 
\[
\frac{\operatorname{Var}[g(X_{\tau_{k_{0}}})]}{n_{0}}+\sum_{l=1}^{L}\frac{\gamma_{k_{l-1}}^{\alpha/2}}{n_l}.
\]
Furthermore, under assumption (AC) the cost of \ensuremath{V_{0}^{\mathbf{n},\mathbf{k}}}
 is bounded from above by a multiple of 
\[
\sum\limits_{l=0}^L (k_l^{\varkappa_1+1} + n_l\cdot k_l^{\varkappa_2})
\]
\end{thm}
Finally, the complexity of \ensuremath{V_{0}^{\mathbf{n},\mathbf{k}}} is given by the following theorem.
\begin{thm}
\label{compl_ml} Let assumptions (AP), (AC), (AQ) and (AM)  hold with
\[
\gamma_{k_l}=k_l^{-\mu},\quad k_l\in\mathbb{N}
\]
for some \ensuremath{\mu>0.}
 Then under the choice \ensuremath{k^*_{l}=k_{0}\cdot \theta^{l},}
 \ensuremath{l=0,1,\ldots,L,}
 with  \(\theta>1,\)
$$L= \left\lceil\frac{2}{\mu(1+\alpha)}\log_\theta \left(\varepsilon^{-1}\cdot k_0^{-\mu(1+\alpha)/2}\right)\right\rceil$$
and
$$n^*_l=\varepsilon^{-2}\left(\sum\limits_{i=1}^L\sqrt{k_i^{(\varkappa_2-\mu\alpha/2)}}\right)\cdot\sqrt{k_l^{(-\varkappa_2-\mu\alpha/2)}}$$
the complexity of the estimate \eqref{mc}
is bounded, up to a constant, from above by 
\begin{equation}
\label{eq:mlmc_comp}
\mathcal{C}_{\mathbf{n}^*,\mathbf{k}^*}(\varepsilon)\lesssim
\begin{cases}
\varepsilon^{-2\cdot\max\left(\frac{\varkappa_1+1}{\mu(1+\alpha)},1\right)}, & 2\cdot\varkappa_2<\mu\alpha\\
\varepsilon^{-2\cdot\frac{\varkappa_1+1}{\mu(1+\alpha)}}, & 2\cdot\varkappa_2=\mu\alpha \text{ and } \frac{\varkappa_1+1}{\mu(1+\alpha)}>1\\
\varepsilon^{-2}\cdot\left(\log\varepsilon\right)^2, & 2\cdot\varkappa_2=\mu\alpha \text{ and } \frac{\varkappa_1+1}{\mu(1+\alpha)}\le1\\
\varepsilon^{-2\cdot\max\left(\frac{\varkappa_1+1}{\mu(1+\alpha)},1+\frac{\varkappa_2-\mu\alpha/2}{\mu(1+\alpha)}\right)}, & 2\cdot\varkappa_2>\mu\alpha\\
\end{cases}
\end{equation}
\end{thm}

\paragraph{Discussion} Let us compare the complexities of the estimates \( V^{n^*,k^*}_0 \) and \(  V^{\mathbf{n}^*,\mathbf{k}^*}_0.  \)
For the sake of clarity we will assume that $\varkappa_1=\varkappa_2=\varkappa$ as in the mesh or local regression methods. Then \eqref{eq:mlmc_comp} versus \eqref{eq:comp_smc} can be written as 
\begin{equation*}
\begin{cases}
\varepsilon^{-2\cdot\max\left(\frac{\varkappa+1}{\mu(1+\alpha)},1\right)}, & 2\cdot\varkappa<\mu\alpha\\
\varepsilon^{-2\cdot\frac{\varkappa+1}{\mu(1+\alpha)}}, & 2\cdot\varkappa=\mu\alpha \text{ and } \frac{\varkappa+1}{\mu(1+\alpha)}>1\\
\varepsilon^{-2}\cdot\left(\log\varepsilon\right)^2, & 2\cdot\varkappa=\mu\alpha \text{ and } \frac{\varkappa+1}{\mu(1+\alpha)}\le1\\
\varepsilon^{-2\cdot\max\left(\frac{\varkappa+1}{\mu(1+\alpha)},1+\frac{\varkappa-\mu\alpha/2}{\mu(1+\alpha)}\right)}, & 2\cdot\varkappa>\mu\alpha\\
\end{cases} 
\ \vee \ 
\varepsilon^{-2\cdot\max\left(\frac{\varkappa+1}{\mu(1+\alpha)},\, 1+\frac{\varkappa}{\mu(1+\alpha)}\right)}
\end{equation*}
Now it is clear that multilevel algorithm will not be superior to the standard Monte Carlo algorithm in the case $\mu(1+\alpha)\le 1$. In the case $\mu(1+\alpha)>1$, the computational gain, up to a logarithmic factor, is given by
$$\begin{cases}
\varepsilon^{-2\cdot\min\left(\frac{\varkappa}{\mu(1+\alpha)},1-\frac{1}{\mu(1+\alpha)}\right)}, & 2\cdot\varkappa<\mu\alpha\\
\varepsilon^{-2\cdot\left(1-\frac{1}{\mu(1+\alpha)}\right)}, & 2\cdot\varkappa=\mu\alpha \text{ and } \frac{\varkappa+1}{\mu(1+\alpha)}>1\\
\varepsilon^{-2\cdot\frac{\varkappa}{\mu(1+\alpha)}}, & 2\cdot\varkappa=\mu\alpha \text{ and } \frac{\varkappa+1}{\mu(1+\alpha)}\le1\\
\varepsilon^{-2\cdot\min\left(1-\frac{1}{\mu(1+\alpha)},\frac{\mu\alpha/2}{\mu(1+\alpha)}\right)}, & 2\cdot\varkappa>\mu\alpha\\
\end{cases} $$
Taking into account the fact that \( \alpha=1 \) in the usual situation, we conclude that it is advantageous to use MLMC as long as \( \mu>1/2. \) 
\section{Numerical example: Bermudan max calls on multiple assets}
\label{num_example}
Suppose that the price of the underlying asset \( X=(X^1,\ldots,X^d) \) follows  a Geometric Brownian motion (GBM) under the risk-neutral measure, i.e.,
\begin{eqnarray}
\label{sde}
dX^i_t=(r-\delta)X^i_t  dt+\sigma X^i_t dB^i_t,
\end{eqnarray}
where \( r \) is the risk-free interest
rate, \( \delta \) the dividend rate, \( \sigma \) the volatility, and \( B_t=(B^1_t,\ldots,B^d_t) \) is a vector of \( d \) independent standard Brownian motions. At any time $t\in\{t_{0}%
,...,t_{\mathcal{J}}\}$ the holder of the option may exercise it and receive
the payoff
\[
h(X_{t})=e^{-rt}(\max(X_{t}^{1},...,X_{t}^{d})-\kappa)^{+}.
\]
We consider a benchmark example (see, e.g. \citet{BG4}, p. 462) when  \( d=5, \) \(\sigma=0.2,\) $r=0.05$,  $\delta=0.1$,  $\kappa=100$
$t_{j}=jT/\mathcal{J},\,j=0,\ldots,\mathcal{J}$, with $T=3$ and
$\mathcal{J}=3$. 
\subsection{Mesh method}
 First note that for the mesh method the conditions of Theorem~\ref{compl_mc} and Theorem~\ref{compl_ml}  are fulfilled with \(\gamma_k=1/k\) in (AQ) and \(\kappa_1=\kappa_2=1\) in (AC). Moreover, for the problem at hand, the assumption (AB) holds with \(\alpha=1\).  Consider  the standard MC mesh approach. For any \(\varepsilon>0\) we set
\begin{eqnarray*}
k = (\varepsilon/2.4)^{-1}, \quad n=(\varepsilon/2.4)^{-2}
\end{eqnarray*}
and simulate independently \( k \) training paths of the process \( Z \)  using the exact formula 
\[  Z_{j}^{(i)}=Z_{{j-1}}^{(i)} \exp \left( \left[r-\delta-\frac{1}{2} \sigma^2\right](t_j-t_{j-1}) +
\sigma\sqrt {(t_j-t_{j-1})} \cdot \xi_j^i \right),   
\]
where \( \xi^i_j, \) \( i=1,\ldots, k, \) are i. i. d. standard normal random variables.
The conditional density of \(Z_j\) given \(Z_{j-1}\) is given by
\[  
p_j(x,y)=\prod_{i=1}^d p_j(x_i,y_i), \quad x=(x_1,\ldots,x_d), \quad y=(y_1,\ldots, y_d), 
\]
where
\begin{eqnarray*}  
p_j(x_i,y_i)&=&\frac{x_i}{y_i\sigma\sqrt{2\pi (t_j-t_{j-1})}}\times
\\
	&& \times\exp\left(\frac{
		-\left(\log\left(\frac{y_i}{x_i}\right)-\left(r-\delta-\frac{1}{2}
		\sigma^2\right)(t_j-t_{j-1})\right)^2}{
		2\sigma^2 (t_j-t_{j-1})}\right). 
\end{eqnarray*}
Using the above paths we construct the sequence of the estimates (training phase)
\[
C_{k,0}(x),\ldots,C_{k,\mathcal{J}}(x)
\] 
as described in Example~\ref{mesh_ex} and then  in testing phase compute the estimate \(V_{0}^{n,k}\) via \eqref{mc}.
Note that for the variance reduction we use inner and outer control variates based on the analytical formula for the European max-call option (see \citet{BG4}) 
\[
\mathcal{E}(x,t,T):=E\left[e^{-rT} \max_{k=1,\ldots,d}\left(X_T^k-\kappa\right)^+\Big| X_t=x\right].
\]
Finally we approximate the mean square error (MSE) of the estimate \(V_{0}^{n,k}\) based on \(100\) repetitions of the training and testing phases. The plot of the estimated quotient \(\sqrt{\text{MSE}}/\varepsilon\) is shown on the l.h.s. of Figure~\ref{fig:mesh}. Turn now to the ML approach. 
Here we take \(\mathbf{n}=(n_0, n_1,\ldots,n_l)\) and \(\mathbf{k}=(k_0, k_1,\ldots,k_l)\) with \(k_0=5\)
\[
k_l=k_0\cdot 2^l, \quad 	n_l=\frac{1}{(\varepsilon/8)^2} \left(\sum_{i=1}^L\sqrt{k_i^{1/2}}\right ) \sqrt{k_l^{-3/2}}, \quad l=0,\ldots,L,
\]
and 
\[
	L=\left\lceil\log_\theta\left(8 \cdot k_0/\varepsilon\right)  \right\rceil.
\] 
The grid for \(\varepsilon\) on the r.h.s. of Figure~\ref{fig:mesh} is chosen in such a way that \(L(\varepsilon)\) runs through the set \(\{1,2,\ldots,7\}.\)
The plot of the estimated quotient \(\sqrt{\text{MSE}}/\varepsilon\) is shown on the r.h.s. of Figure~\ref{fig:mesh}.
Figure~\ref{fig:mesh} suggests that the rates given  in Theorem~\ref{thm:bias} and Theorem~\ref{compl_ml} do hold.
Next we compare the computational cost  
\begin{eqnarray*}
\sum\limits_{l=0}^L (k_l^2 + n_l\cdot k_l)
\end{eqnarray*}
where \(k_l,\)  \(n_l\) and \(L\) are defined above to the theoretical complexity given by \(\varepsilon^{-2.5}.\)
In Figure~\ref{fig:comp} we present the corresponding log-plots of complexities and the gains as functions of \(\log (1/\varepsilon).\) 
\begin{center}
\begin{figure}
\label{fig:mesh}
\includegraphics[width=\textwidth]{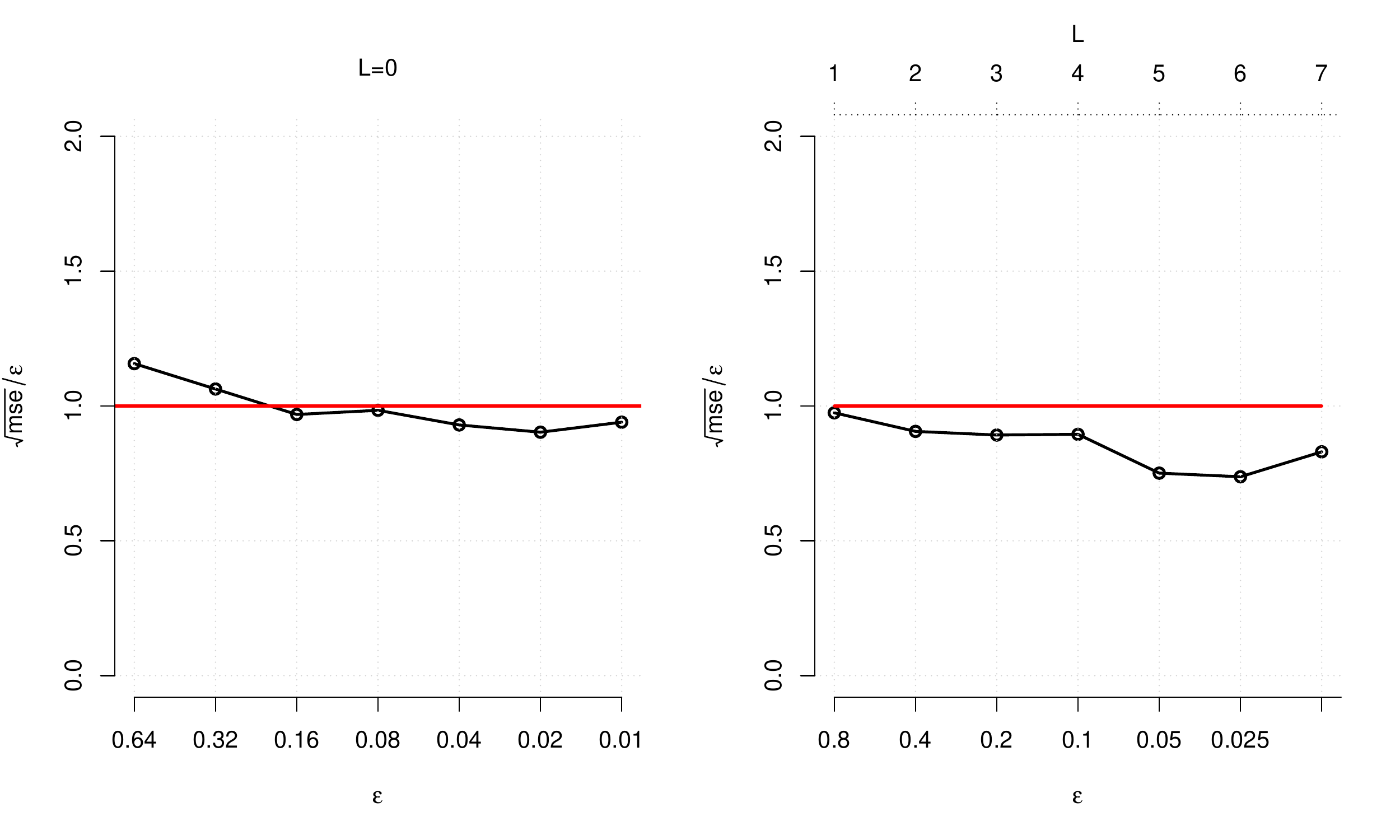}
\caption{Mesh method: mean square errors of the standard Monte Carlo estimate \(V_{0}^{n,k}\) (left) and the multilevel estimate \(V_{0}^{\mathbf{n},\mathbf{k}}\) (right) in the units of the expected error \(\varepsilon.\)}  
\end{figure}
\end{center}
\begin{center}
\begin{figure}
\includegraphics[width=\textwidth]{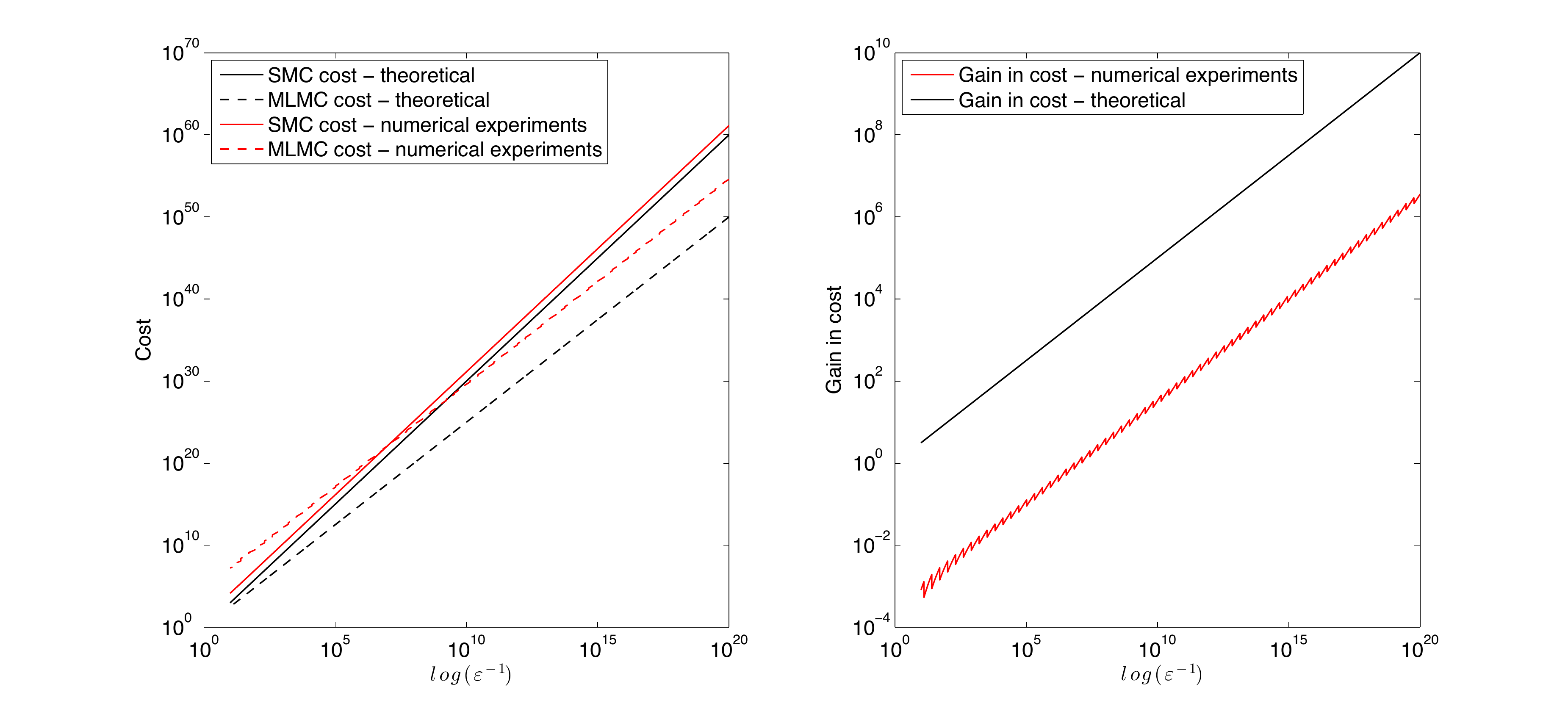}
\caption{Comparison of theoretical and numerical costs  (left) and theoretical and numerical gains (right) as functions  of \(\log(\varepsilon^{-1})\) \label{fig:comp}
}  
\end{figure}
\end{center}

\subsection{Local regression}
We use the local constant regression with the simplest kernels of the form:
\begin{eqnarray*}
K(z)=1(|z|\leq 1)
\end{eqnarray*}
and define
\begin{eqnarray}
\label{local_est}
C_{k,j}(z)=\sum_{i=1}^{k}\zeta_{k,j+1}(Z_{j+1}^{(i)})\cdot w^k_{ij}(z),
\end{eqnarray}
where 
\begin{eqnarray*}
w^k_{ij}(z)=\frac{1\bigl(|z-Z_{j}^{(i)}|\leq \delta_k\bigr)}{\sum_{l=1}^{k} 1\bigl(|z-Z_{j}^{(l)}|\leq \delta_k\bigr)}, \quad i=1,\ldots, k,
\end{eqnarray*}
with \(\delta_k=100\cdot k^{-1/(d+2)}\) (see Example~\ref{mesh_ex}).
For any \(\varepsilon>0\) we set
\begin{eqnarray*}
k = (\varepsilon/1.2)^{-6}, \quad n=(\varepsilon/1.2)^{-2}
\end{eqnarray*}
corresponding to the choice \(\gamma_k=k^{-1/6}\) (\(\mu=1/6\)) in Theorem~\ref{compl_mc} and approximate the mean square error (MSE) of the MC estimate \(V_{0}^{n,k}\) based on \(100\) repetitions of the training and testing phases.
In the case of the MLMC algorithm we take $k_0=100,$ 
\[
	L=\left\lceil6\cdot\log_\theta\left( \frac{3}{\varepsilon \cdot k_0^{1/6}}\right)  \right\rceil,
\]
and
\[
k_l=k_0 2^l, \quad n_l=\frac{10}{(\varepsilon/3)^2} \left(\sum_{i=1}^L\sqrt{(k_i)^{11/12}}\right) \sqrt{(k_l)^{-13/12}},\quad l=0,\ldots,L.
\]
The results in form of the quotients \(\sqrt{MSE}/\varepsilon\) are shown in Figure~\ref{fig:local}.
\begin{center}
\begin{figure}
\label{fig:local}
  \includegraphics[width=\textwidth]{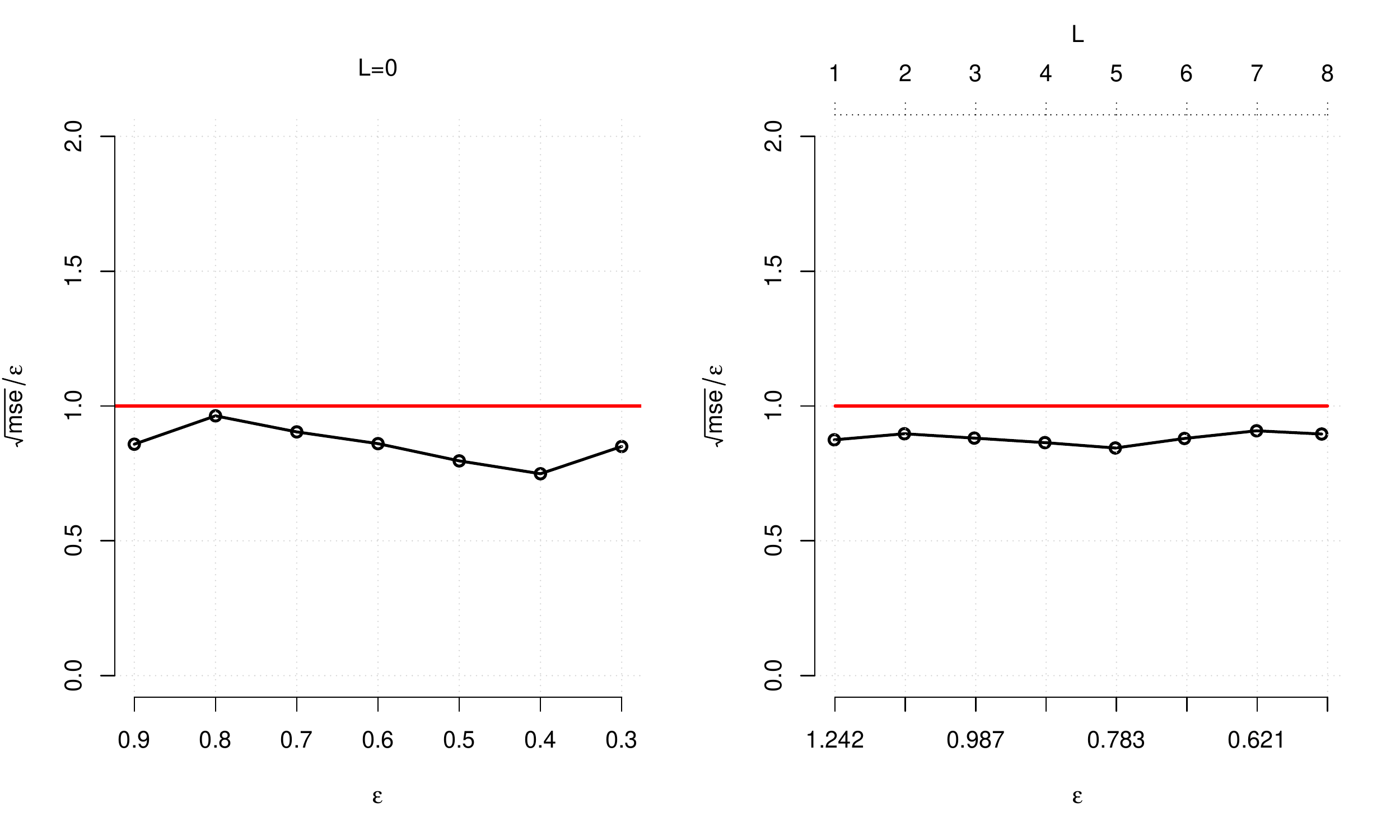}
\caption{Local regression: mean square errors of the standard Monte Carlo estimate \(V_{0}^{n,k}\) (left) and the multilevel estimate \(V_{0}^{\mathbf{n},\mathbf{k}}\) (right) in the units of the expected error \(\varepsilon.\)}    
\end{figure}
\end{center}
\section{Proofs}
\label{proofs}
\subsection{Proof of Theorem \ref{thm:bias}}

A family of stopping times $\left(\tau_{j}\right)_{j=0,\ldots,\mathcal{J}}$
w.r.t. the filtration $(\mathcal{F}_{j})_{j=0,\ldots,\mathcal{J}}$
is called consistent if 
\[
j\leq\tau_{j}\leq\mathcal{J},\quad\tau_{\mathcal{J}}=\mathcal{J}
\]
and
\[
\tau_{j}>j\quad\Longrightarrow\quad\tau_{j}=\tau_{j+1}.
\]

\begin{lem}
\label{lem:haupt_techn}Let $(Y_{j})_{j=0,\ldots,\mathcal{J}}$ be
a process adapted to the filtration $(\mathcal{F}_{j})_{j=0,\ldots,\mathcal{J}}$
and let $\left(\tau_{j}^{1}\right)$ and $\left(\tau_{j}^{2}\right)$
be two consistent families of stopping times. Then 
\[
\mathrm{E}^{\mathcal{F}_{j}}\left[Y_{\tau_{j}^{1}}-Y_{\tau_{j}^{2}}\right]=\mathrm{E}^{\mathcal{F}_{j}}\left\{ \sum_{l=j}^{\mathcal{J}-1}\left(Y_{l}-\mathrm{E}^{\mathcal{F}_{l}}\left[Y_{\tau_{l+1}^{1}}\right]\right)\left(1_{\{\tau_{l}^{1}=l,\tau_{l}^{2}>l\}}-1_{\{\tau_{l}^{1}>l,\tau_{l}^{2}=l\}}\right)1_{\{\tau_{l}^{2}>l\}}\right\} 
\]
for any $j=0,\ldots,\mathcal{J}-1.$ 
\end{lem}
\begin{proof}
We have
\[
\]
\begin{eqnarray*}
Y_{\tau_{j}^{1}}-Y_{\tau_{j}^{2}} & = & \left[Y_{j}-Y_{\tau_{j}^{2}}\right]1_{\{\tau_{j}^{1}=j,\tau_{j}^{2}>j\}}+\left[Y_{\tau_{j}^{1}}-Y_{j}\right]1_{\{\tau_{j}^{1}>j,\tau_{j}^{2}=j\}}\\
 &  & +\left[Y_{\tau_{j}^{1}}-Y_{\tau_{j}^{2}}\right]1_{\{\tau_{j}^{1}>j,\tau_{j}^{2}>j\}}\\
 & = & \left[Y_{j}-Y_{\tau_{j+1}^{1}}\right]1_{\{\tau_{j}^{1}=j,\tau_{j}^{2}>j\}}+\left[Y_{\tau_{j+1}^{1}}-Y_{j}\right]1_{\{\tau_{j}^{1}>j,\tau_{j}^{2}=j\}}\\
 &  & +\left[Y_{\tau_{j+1}^{1}}-Y_{\tau_{j+1}^{2}}\right]1_{\{\tau_{j}^{1}=j,\tau_{j}^{2}>j\}}+\left[Y_{\tau_{j+1}^{1}}-Y_{\tau_{j+1}^{2}}\right]1_{\{\tau_{j}^{1}>j,\tau_{j}^{2}>j\}}.
\end{eqnarray*}
Therefore it holds for $\Delta_{j}=\mathrm{E}^{\mathcal{F}_{j}}\left[Y_{\tau_{j}^{1}}-Y_{\tau_{j}^{2}}\right]$
\[
\Delta_{j}=\left[Y_{j}-\mathrm{E}^{\mathcal{F}_{j}}\left[Y_{\tau_{j+1}^{1}}\right]\right]\left(1_{\{\tau_{j}^{1}=j,\tau_{j}^{2}>j\}}-1_{\{\tau_{j}^{1}>j,\tau_{j}^{2}=j\}}\right)+\mathrm{E}^{\mathcal{F}_{j}}\left\{ \Delta_{j+1}1_{\{\tau_{j}^{2}>j\}}\right\} 
\]
with $\Delta_{\mathcal{J}}=0$ and 
\[
\Delta_{j}=\mathrm{E}^{\mathcal{F}_{j}}\left\{ \sum_{l=j}^{\mathcal{J}-1}\left(Y_{l}-\mathrm{E}^{\mathcal{F}_{l}}\left[Y_{\tau_{l+1}^{1}}\right]\right)\left(1_{\{\tau_{l}^{1}=l,\tau_{l}^{2}>l\}}-1_{\{\tau_{l}^{1}>l,\tau_{l}^{2}=l\}}\right)1_{\{\tau_{l}^{2}>l\}}\right\} .
\]
\end{proof}
Taking into account that 
\[
C_{l}^{*}(Z_{l})=\mathrm{E}^{\mathcal{F}_{l}}\left[g_{\tau_{l+1}^{*}}(Z_{\tau_{l+1}^{*}})\right]\leq g_{l}(Z_{l})
\]
on $\{\tau_{l}^{*}=l\}$ and 
\[
C_{l}^{*}(Z_{l})< g_{l}(Z_{l})
\]
on $\{\tau_{l}^{*}>l\},$ we get from Lemma \ref{lem:haupt_techn} for \(R=V_0^{n,k}-V_0^*\)
\begin{eqnarray*}
|R| & = & \left|\mathrm{E}\left[g_{\tau_{k}^{*}}(Z_{\tau_{k}^{*}})-g_{\tau_{k}}(Z_{\tau_{k}})\right]\right|\\
 & \leq & \mathrm{E}\left[\sum_{l=0}^{\mathcal{J}-1}\left|C_{l}^{*}(Z_{l})-g_{l}(Z_{l})\right|\left(1_{\{\tau^*_{k,l}=l,\tau_{k,l}>l\}}+1_{\{\tau^*_{k,l}>l,\tau_{k,l}=l\}}\right)\right].
\end{eqnarray*}
Introduce 
\begin{align*}
\mathcal{E}_{k,j} & =\{g_{j}(Z_{j})>C_{k,j}^{*}(Z_{j}),\: g_{j}(Z_{j})\leq C_{k,j}(Z_{j})\}\\
 & \cup\{g_{j}(Z_{j})\leq C_{k}^{*}(Z_{j}),\: g_{j}(Z_{j})>C_{k,j}(Z_{j})\},
\end{align*}
\[
\mathcal{A}_{k,j,0}=\left\{ 0<\left|g_{j}(Z_{j})-C_{j}^{*}(Z_{j})\right|\leq\gamma_{k}^{-1/2}\right\} ,
\]
\[
\mathcal{A}_{k,j,i}=\left\{ 2^{i-1}\gamma_{k}^{-1/2}<\left|g_{j}(Z_{j})-C_{j}^{*}(Z_{j})\right|\leq2^{i}\gamma_{k}^{-1/2}\right\} 
\]
for $j=0,\ldots,\mathcal{J}-1$ and $i>0.$ It holds
\begin{eqnarray*}
|R| & \leq & \mathrm{E}\left[\sum_{l=0}^{\mathcal{J}-1}\left|C_{l}^{*}(Z_{l})-g_{l}(Z_{l})\right|1_{\{\mathcal{E}_{k,l}\}}\right]\\
 & = & \mathrm{E}\left[\sum_{i=0}^{\infty}\sum_{l=0}^{\mathcal{J}-1}\left|C_{l}^{*}(Z_{l})-g_{l}(Z_{l})\right|1_{\{\mathcal{E}_{k,l}\cap\mathcal{A}_{k,l,i}\}}\right]\\
 & =\gamma_{k}^{-1/2} & \sum_{l=0}^{\mathcal{J}-1}\mathrm{P}\left(\left|g_{l}(Z_{l})-C_{l}^{*}(Z_{l})\right|\leq\gamma_{k}^{-1/2}\right)\\
 &  & +\mathrm{E}\left[\sum_{i=1}^{\infty}\sum_{l=0}^{\mathcal{J}-1}\left|C_{l}^{*}(Z_{l})-g_{l}(Z_{l})\right|1_{\{\mathcal{E}_{k,l}\cap\mathcal{A}_{k,l,i}\}}\right].
 \end{eqnarray*}
Using the fact that $\left|g_{l}(Z_{l})-C_{l}^{*}(Z_{l})\right|\leq\left|C_{l}(Z_{l})-C_{l}^{*}(Z_{l})\right|$
on $\mathcal{E}_{k,l},$ we derive 
\begin{eqnarray*}
|R| & \leq & \gamma_{k}^{-1/2}\sum_{l=0}^{\mathcal{J}-1}\mathrm{P}\left(\left|g_{l}(Z_{l})-C_{l}^{*}(Z_{l})\right|\leq\gamma_{k}^{-1/2}\right)\\
 &  & +\sum_{i=1}^{\infty}2^{i}\gamma_{k}^{-1/2}\mathrm{E}\left[\sum_{l=0}^{\mathcal{J}-1}1_{\left\{ \left|g_{j}(Z_{j})-C_{j}^{*}(Z_{j})\right|\leq2^{i}\gamma_{k}^{-1/2}\right\} }\mathrm{P}^{k}\left(\left|C_{k,l}(Z_{l})-C_{l}^{*}(Z_{l})\right|>2^{i-1}\gamma_{k}^{-1/2}\right)\right]\\
 & \leq & A\mathcal{J}\gamma_{k}^{-\alpha/2}+A\mathcal{J}\gamma_{k}^{-\alpha/2}\sum_{i=1}^{\infty}2^{i}B_{1}\exp(-B_{2}2^{i-1}).
\end{eqnarray*}
\subsection{Proof of Theorem \ref{compl_mc}}
Based on \eqref{eq:var_smc} we have the optimization problem
\begin{gather*}
k^{\varkappa_1+1} +n\cdot k^{\varkappa_2}\to\min\\
\gamma_{k}^{(1+\alpha)/2}\le\varepsilon\\
n\ge\varepsilon^{-2}
\end{gather*}
It is clear that
$$\gamma_{k}^{(1+\alpha)/2}=k^{-\mu(1+\alpha)/2}\Rightarrow k\ge \varepsilon^{-\frac{2}{\mu(1+\alpha)}},$$
which immediately leads to the  statement.
\subsection{Proof of Theorem \ref{biasvarml}}
The formula for the variance follows from the estimate
\begin{eqnarray*}
\mathrm{E}\left[g_{\tau_{k_{l}}}\Bigl(Z_{\tau_{k_{l}}}\Bigr)-g_{\tau_{k_{l-1}}}\Bigl(Z_{\tau_{k_{l-1}}}\Bigr)\right]^2 
&\leq &
\mathrm{E}\left[g_{\tau^*}\Bigl(Z_{\tau^*}\Bigr)-g_{\tau_{k_{l-1}}}\Bigl(Z_{\tau_{k_{l-1}}}\Bigr)\right]^2
\\
&&+\mathrm{E}\left[g_{\tau_{k_{l}}}\Bigl(Z_{\tau_{k_{l}}}\Bigr)-g_{\tau^*}\Bigl(Z_{\tau^*}\Bigr)\right]^2
\\
&\leq & 2^{\mathcal{J}}G^2\sum_{l=0}^{\mathcal{J}-1}\Bigl[\mathrm{P}(\mathcal{E}_{k_{l-1},l})+\mathrm{P}(\mathcal{E}_{k_{l},l})\Bigr],
\end{eqnarray*}
where for any \(k\)
\begin{eqnarray*}
\sum_{l=0}^{\mathcal{J}-1}\mathrm{P}(\mathcal{E}_{k,l})& \leq & \sum_{l=0}^{\mathcal{J}-1}\mathrm{P}\left(\left|g_{l}(Z_{l})-C_{l}^{*}(Z_{l})\right|\leq\gamma_{k}^{-1/2}\right)\\
 &  & +\sum_{i=1}^{\infty}2^{i}\mathrm{E}\left[\sum_{l=0}^{\mathcal{J}-1}1_{\left\{ \left|g_{j}(Z_{j})-C_{j}^{*}(Z_{j})\right|\leq2^{i}\gamma_{k}^{-1/2}\right\} }\mathrm{P}^{k}\left(\left|C_{k,l}(Z_{l})-C_{l}^{*}(Z_{l})\right|>2^{i-1}\gamma_{k}^{-1/2}\right)\right]\\
 & \leq & A\mathcal{J}\gamma_{k}^{-\alpha/2}+A\mathcal{J}\gamma_{k}^{-\alpha/2}\sum_{i=1}^{\infty}2^{i}B_{1}\exp(-B_{2}2^{i-1}).
\end{eqnarray*}

\subsection{Proof of Theorem \ref{compl_ml}}
Due to the monotone structure of the functional, we can consider the following optimization problem:

\begin{gather}
\label{eq:mlmc1}
\sum\limits_{l=0}^L k_l^{\varkappa_1+1} + n_l\cdot k_l^{\varkappa_2}\to\min\\
\label{eq:mlmc2}
\gamma_{k_L}^{(1+\alpha)/2}=k_L^{-\mu(1+\alpha)/2}=\left(k_0\cdot\theta^{L}\right)^{-\mu(1+\alpha)/2}\le\varepsilon\\
\label{eq:mlmc3}
\frac{1}{n_0}     +    \sum\limits_{l=1}^L \frac{\gamma^{\alpha/2}_{k_{l-1}}}{n_l}\asymp
k_0^{-\mu\alpha/2}\cdot\sum\limits_{l=0}^L \frac{\theta^{-l\mu\alpha/2}}{n_l}=\varepsilon^{2}
\end{gather}
Now the Lagrange multiplier method with respect to $n_l$ gives us
$$ k_l^{\varkappa_2}=-\lambda\frac{k^{-\mu\alpha/2}_l}{n_l^2}\Rightarrow n_l = \sqrt{(-\lambda)\cdot k_l^{(-\varkappa_2-\mu\alpha/2)}}.$$
Now one can put the value of $n_l$ in \eqref{eq:mlmc3}:
$$\sum\limits_{l=1}^L \frac{\gamma^{\alpha/2}_{k_{l-1}}}{n_l} \asymp \sum\limits_{l=1}^L \frac{k_l^{-\mu\alpha/2}}{\sqrt{(-\lambda)\cdot k_l^{(-\varkappa_2-\mu\alpha/2)}}}\asymp\varepsilon^{2}$$
$$\Downarrow$$
$$\sqrt{(-\lambda)}=\varepsilon^{-2}\cdot \sum\limits_{l=1}^L\sqrt{k_l^{(\varkappa_2-\mu\alpha/2)}}$$
$$\Downarrow$$
$$n_l=\varepsilon^{-2}\left(\sum\limits_{i=1}^L\sqrt{k_i^{(\varkappa_2-\mu\alpha/2)}}\right)\cdot\sqrt{k_l^{(-\varkappa_2-\mu\alpha/2)}}.$$
For total number of level we have from \eqref{eq:mlmc2}:
$$\left(k_0\cdot\theta^{L}\right)^{-\mu(1+\alpha)/2}\le\varepsilon\Rightarrow L\ge \frac{2}{\mu(1+\alpha)}\log_\theta \left(\varepsilon^{-1}\cdot k_0^{-\mu(1+\alpha)/2}\right).$$
Now we can rewrite \eqref{eq:mlmc1} as
$$\sum\limits_{l=0}^L k_l^{\varkappa_1+1} + n_l\cdot k_l^{\varkappa_2}\asymp k_L^{\varkappa_1+1}+\varepsilon^{-2}\cdot\left(\sum\limits_{l=1}^L\sqrt{k_l^{(\varkappa_2-\mu\alpha/2)}}\right)^2,$$
so we will have three cases.
\begin{enumerate}
\item[Case 1.] $2\cdot\varkappa_2=\mu\alpha.$
\begin{align*}
k_L^{\varkappa_1+1}+\varepsilon^{-2}\cdot\left(\sum\limits_{l=1}^L\sqrt{k_l^{(\varkappa_2-\mu\alpha/2)}}\right)^2 & \gtrsim k_L^{\varkappa_1+1}+\varepsilon^{-2}\cdot L^2 \\
& \geq\varepsilon^{-\frac{2\cdot(\varkappa_1+1)}{\mu(1+\alpha)}} +\varepsilon^{-2}\cdot L^2
\end{align*}
\item[Case 2.] $2\cdot\varkappa_2<\mu\alpha.$
\begin{align*}
k_L^{\varkappa_1+1}+\varepsilon^{-2}\cdot\left(\sum\limits_{l=1}^L\sqrt{k_l^{(\varkappa_2-\mu\alpha/2)}}\right)^2 &\gtrsim k_L^{\varkappa_1+1}+\varepsilon^{-2}\\ 
&\geq\varepsilon^{-\frac{2\cdot(\varkappa_1+1)}{\mu(1+\alpha)}} +\varepsilon^{-2}
\end{align*}
\item[Case 3.]   $2\cdot\varkappa_2>\mu\alpha.$
\begin{align*}
k_L^{\varkappa_1+1}+\varepsilon^{-2}\cdot\left(\sum\limits_{l=1}^L\sqrt{k_l^{(\varkappa_2-\mu\alpha/2)}}\right)^2 & \gtrsim k_L^{\varkappa_1+1}+\varepsilon^{-2}\cdot k_L^{\varkappa_2-\mu\alpha/2}\\ 
&\geq \varepsilon^{-\frac{2\cdot(\varkappa_1+1)}{\mu(1+\alpha)}} +\varepsilon^{-2-\frac{2\varkappa_2-\mu\alpha}{\mu(1+\alpha)}}
\end{align*}
\end{enumerate}
Combining all three cases one will get \eqref{eq:mlmc_comp}.

\end{document}